\documentclass[twoside,11pt]{article}

\usepackage{makeidx}  
\usepackage{algorithm}
\usepackage{algorithmicx,algpseudocode}
\usepackage{amsmath,amsfonts,mathrsfs}
\usepackage{xspace}
\usepackage{graphicx}
\usepackage{epsfig}
\usepackage{authblk}

\newtheorem{theorem}{Theorem} 
\newtheorem{lemma}{Lemma}

  \newtheorem{definition}{Definition} 
  \newtheorem{remarka}{Remark}

 \newenvironment{proof}{{\bf Proof.}}{\hfill\rule{2mm}{2mm}} 
 \newenvironment{pproof}[1]{\noindent{\textbf{Proof of #1.}}}{\hfill\rule{2mm}{2mm}}

\newcommand{\oi}{{\hat{i}}}
\newcommand{\cent}{{\rm cent}\xspace}
\newcommand{\cost}{{\rm cost}\xspace}
\newcommand{\mR}{{\mathcal{R}}}
\newcommand{\mB}{{\mathcal{B}}}

\newcommand{\rbm}{{\textsc{Budgeted Red-Blue Median}}\xspace}
\newcommand{\km}{{\textsc{$k$-Median}}\xspace}
\newcommand{\mm}{{\textsc{Matroid Median}}\xspace}
\newcommand{\ufl}{{\textsc{Uncapacitated Facility Location}}\xspace}
\newcommand{\cfl}{{\textsc{Capacitated Facility Location}}\xspace}

\title{Tight Analysis of a Multiple-Swap Heuristic for Budgeted Red-Blue Median}
\author{Zachary Friggstad\thanks{This research was undertaken, in part, thanks to funding from the Canada Research Chairs program and an NSERC Discovery Grant.} \qquad Yifeng Zhang\\
Department of Computing Science\\
University of Alberta\\
\texttt{\{zacharyf,yifeng2\}@ualberta.ca}}
\date{}

\begin{document}
\maketitle

\begin{abstract}
\rbm is a generalization of classic \km in that there are two sets of facilities, say $\mR$ and $\mB$, that can be used to serve clients located in some metric space.
The goal is to open $k_r$ facilities in $\mR$ and $k_b$ facilities in $\mB$ for some given bounds $k_r, k_b$
and connect each client to their nearest open facility in a way that minimizes the total connection cost.

We extend work by Hajiaghayi, Khandekar, and Kortsarz [2012] and show that a {\em multiple-swap} local search heuristic can be used to obtain a
$(5+\epsilon)$-approximation for \rbm for any constant $\epsilon > 0$. This is an improvement over their single swap analysis and beats the previous
best approximation guarantee of 8 by Swamy [2014].

We also present a matching lower bound showing that for every $p \geq 1$, there are instances of \rbm with local optimum solutions for the $p$-swap
heuristic whose cost is $5 + \Omega\left(\frac{1}{p}\right)$ times the optimum solution cost. Thus, our analysis is tight up to the lower order terms.
In particular, for any $\epsilon > 0$ we show the single-swap heuristic admits local optima whose cost
can be as bad as $7-\epsilon$ times the optimum solution cost.
\end{abstract}



\section{Introduction} \label{sec:intro}

Facility location problems crop up in many areas of computing science and operations research. A typical problem involves a set of clients
and possible facility locations located in a metric space. The goal is to open some facilities and connect each client to some open facility as cheaply as possible.
These problems become difficult when there are costs associated with opening facilities or additional constraints that ensure we cannot open too many facilities.

We study \rbm, one particular instance of this type of problem.
Here we are given a set of clients $C$, a set of {\em red} facilities $\mR$, and a set of {\em blue} facilities $\mB$. These are located in some metric space
with metric distances $d(i,j) \geq 0$ for any two $i,j \in C \cup \mR \cup \mB$. Additionally, we are given two integer bounds $k_r \leq |\mR|$ and $k_b \leq |\mB|$.
The goal is to select/open $k_r$ red facilities $R$ and $k_b$ blue facilities $B$ to minimize
\[ \cost(R \cup B) := \sum_{j \in C} \min_{i \in R \cup B} d(i,j). \]
The classic NP-hard \km problem appears as a special case when, say, $\mR = \emptyset$. Thus, \rbm is NP-hard. In this paper, we focus on approximation algorithms
for \rbm, in particular on local search techniques.


\subsection{Previous Work}
The study of \rbm from the perspective of approximation algorithms
was initiated by Hajiaghayi, Khandekar, and Kortsarz \cite{HKK12}, where they obtain a constant-factor approximation by a local search algorithm that iteratively tries to swap one red
and/or one blue facility in the given solution. They do not specify the constant in their analysis, but it looks to be greater than 8.
Citing \cite{HKK12} as inspiration, Krishnaswamy et al. studied a generalization of \rbm known as \mm \cite{KKNSS11}. Here, a matroid structure is given over the set of facilities and we can only
open a set of facilities if they form an independent set in the matroid. They obtain a constant-factor approximation for \mm through rounding an LP relaxation.
This was later refined to an 8-approximation by Swamy \cite{S14}.

The special case of \km is a classic optimization problem and has received a lot of attention from both theoretical and practical communities.
The best approximation guarantee known so far is 2.675 by Byrka et al. \cite{BPRST15}, who build heavily on the breakthrough work
of Li and Svensson for the problem \cite{LS13}.

While local search techniques have been used somewhat infrequently in the design of approximation algorithms in general,
it may be fair to say that they have seen the most success in facility location problems.
For almost 10 years, the best approximation for \km was based on a local search algorithm. Arya et al. \cite{AGKMMP04} show that a multiple-swap heuristic leads
to a $(3+\epsilon)$-approximation for \km for any constant $\epsilon > 0$. This analysis was simplified in \cite{GT08}, which inspires much of our analysis.

Other successful local search applications include a $(1+\sqrt 2)$-approximation for \ufl is also obtained through local search \cite{AGKMMP04,CG05}. 
Local search has been very helpful in approximating \cfl, the first constant-factor approximation was by P\'al, Tardos, and Wexler \cite{PTW01}
and the current best approximation is a $(5+\epsilon)$-approximation by Bansal, Garg, and Gupta \cite{BGG12}. In the special case when
all capacities are uniform, Aggarwal et al. \cite{AABGGGJ10} obtain a 3-approximation.
Even more examples of local search applied to other facility location variants can be found in \cite{AFS13,GN11,GT08,MP03,ST10}.
%
\subsection{Our Results and Techniques}
We show that a multiswap generalization of the local search algorithm considered in \cite{HKK12} is a $(5+\epsilon)$-approximation for \rbm. That is, for a value $p$ say the $p$-swap heuristic is the algorithm
that, upon given an initial feasible solution, tries to swap up to $p$ facilities of each colour. If no such swap produces a cheaper solution, it terminates. Otherwise, it iterates with the now cheaper solution.
The formal description is given in Algorithm \ref{alg:alg_main} in Section \ref{sec:prelim}.

Say that a solution is locally optimum for the $p$-swap heuristic if no cheaper solution can be found by swapping up to $p$ facilities of each colour.
Let $OPT$ denote the cost of an optimum solution. Our main result is the following.
\begin{theorem}\label{thm:main}
Any locally optimum solution for the $p$-swap heuristic has cost at most $(5 + O(1/\sqrt p)) \cdot OPT$.
\end{theorem}
Using standard techniques (briefly mentioned in Section \ref{sec:prelim}), this readily leads to a polynomial-time approximation algorithm.
By choosing $p = \theta(1/\epsilon^2)$ we have the following.
\begin{theorem}
For any constant $\epsilon > 0$, \rbm admits a polynomial-time $(5+\epsilon)$-approximation.
\end{theorem}
This improves over the 8-approximation for \rbm in \cite{S14}.
We emphasize the approximation guarantee from Theorem \ref{thm:main} is for \rbm only,
the 8-approximation in \cite{S14} is still the best approximation for
the general \mm problem. Indeed, \cite{KKNSS11} show that \mm cannot be approximated within any constant factor using any constant number of swaps
even in the generalization of \rbm where there can be a super-constant number facility colours.

We also present a lower bound that matches our analysis up the lower order terms.
\begin{theorem}\label{thm:localgap}
For any integers $p,\ell$ with $1 \leq p \leq \ell/2$, there is an instance of \rbm that has a locally-optimum solution for the $p$-swap heuristic with
cost at least $\left(5 + \frac{2}{p} - \frac{10p}{\ell+1}\right) \cdot OPT$.
\end{theorem}

By letting $\ell \rightarrow \infty$ but keeping $p$ fixed, we see that the $p$-swap heuristic cannot guarantee a ratio better than
$5 + \frac{2}{p}$. So, Theorem \ref{thm:main} is tight up to lower order terms.
Also, for $p = 1$ we see that the single-swap heuristic analyzed in \cite{HKK12} is not better than a 7-approximation.

Local search techniques are typically analyzed by constructing a set of candidate {\em test swaps} where some facilities in the optimum solution are swapped in and some from the local
optimum are swapped out in order to generate a useful inequality.
One of the main features of the \km analysis in \cite{AGKMMP04} and \cite{GT08} is that such swaps can be considered that ensure each facility in the global optimum is swapped in once
and, by averaging some swaps, each facility in the local optimum is swapped out to the extent of at most $1+O(\epsilon)$ times. Each time a facility in the local optimum is swapped out, they pay an additional
2 times the global optimum cost for some clients to reassign them.

We obtain only a $5+\epsilon$ approximation because we end up swapping out some facilities in the local optimum solution to the extent of $2+O(\epsilon)$, thereby paying an additional $2 + O(\epsilon)$
more than in the \km analysis.
Ultimately, this is because some of our initial swaps generate inequalities that depend {\em positively} on client assignment costs in the local optimum.
So we consider additional swaps that do not introduce any more positive dependence on the local optimum to cancel them out.

This issue was also encountered in the analysis in \cite{HKK12}. In some sense, we are showing that this is the only added difficulty over the standard \km analysis.
However, the averaging arguments we use are a fair bit more sophisticated than the analysis for \km.


\subsection{Organization}
Section \ref{sec:prelim} presents the algorithm and describes some useful notation. In particular, it presents a way to decompose the global and
local optimum solution into structured groups that are examined in the analysis. Section \ref{sec:multiswap} analyzes the
quality of locally optimum solutions to prove Theorem \ref{thm:main}. Section \ref{sec:localgap} proves Theorem \ref{thm:localgap}
with an explicit construction of a bad example. We conclude with some remarks in Section \ref{sec:conclusion}.


\section{Notation and Preliminaries}\label{sec:prelim}

Say that a {\em feasible solution} is a pair $(R,B)$ of subsets $R \subseteq \mR$ and $B \subseteq \mB$ with
$|R| = k_r$ and $|B| = k_b$. Algorithm \ref{alg:alg_main} describes the local search algorithm.

\begin{algorithm}
\caption{The $p$-Swap Heuristic for \rbm}\label{alg:alg_main}
\begin{algorithmic}
\State Let $(R,B)$ be arbitrary feasible solution.
\While{there is some feasible solution $(R',B')$ with $|R-R'| \leq p$ \\
\hspace{9mm} and $|B-B'| \leq p$ and $\cost(R' \cup B') < \cost(R \cup B)$}
\State $(R,B) \leftarrow (R',B')$
\EndWhile
\State \Return $(R,B)$
\end{algorithmic}
\end{algorithm}

While a single iteration of Algorithm \ref{alg:alg_main} can be executed in $n^{O(p)}$ (where $n$ is the total number of locations in the problem), it may be that the number
of iterations is not polynomially bounded. We can employ a well-known trick to ensure it does terminate in a polynomial number of steps
while losing only another $\epsilon$ in our analysis.
The idea is to perform the update only if $\cost(R' \cup B') \leq (1 - \epsilon/\Delta) \cdot \cost(R \cup B)$ where $\Delta$ is some quantity that is
polynomial in the input size.
Our analysis is compatible with this approach; one can check that the total weight of all inequalities we consider is polynomially bounded.
For example, see \cite{AGKMMP04} for details. We do not focus any further on this detail, and instead work toward analyzing
the cost of the solutions produced by Algorithm \ref{alg:alg_main} as it is stated.


%
%
 
From now on, let $S = R \cup B$ with $R \subseteq \mR, B \subseteq \mB$ denote an arbitrary local optimum solution. That is, there is no cheaper solution $(R',B')$ with $|R-R'| \leq p$ and $|B-B'| \leq p$.
Also fix a global optimum solution $O = R^* \cup B^*$ where $R^* \subseteq \mB$ and $B^* \subseteq \mB$. We assume that $S \cap O = \emptyset$. This is without loss of generality, as we can duplicate
each facility location in the input and say that $S$ use the first copies and $O$ use the second copies. It is easy to verify that $S$ is still a local optimum solution.

To help analyze the cost, we will introduce some notation.
For any client $j \in \mathcal{C}$, let $s_j \in S$ denote the local optimum facility is closest to $j$ and $o_j \in O$
be the optimum facility that is closest to $j$. For brevity, let $c_j = d(j, s_j)$ be the cost of assigning $j$ in the local optimum and $c^*_j = d(j, o_j)$ the cost of assigning $j$ in the global optimum.
Thus, $\cost(S) = \sum_{j \in C} c_j$ and $\cost(O) = \sum_{j \in C} c^*_j$.
For any facility $i^* \in O$ we let $N^*(i^*) = \{ j \in C : o_j = i^*\}$ and for any $i \in S$ we let $N(i) = \{j \in C : s_j = i\}$.

Let $\phi : O \rightarrow S$ map each facility in $O$ to its nearest facility in $S$, breaking ties arbitrarily.
For $i \in S$, let $\deg(i) = |\phi^{-1}(i)|$. If $\deg(i) \neq 0$, let $\cent(i)$ be the facility in $\phi^{-1}(i)$ that is closest to $i$, again breaking ties arbitrarily.
 
We also borrow some additional notation from \cite{HKK12}.
\begin{definition}[very good, good, bad facility]\label{def:facgood}
A facility $i \in S$ is \textit{very good} if $\deg(i) = 0$, \textit{good} if no $i^* \in \phi^{-1}(i)$ has the same colour as $i$, and \textit{bad} otherwise.
\end{definition}

The analysis in \cite{HKK12} divides $S \cup O$ into {\em blocks} that satisfy certain properties. We require slightly stronger properties than
their blocks guarantee. We also use a slightly different notion of what it means for some $i \in S$ to be a {\em leader}. The required properties
are summarized in the following lemma.

\begin{lemma}\label{lem:block_prop}
We can partition $S \cup O$ into {\em blocks} $T$ satisfying the following properties.
\begin{itemize}
\item $|T \cap R| = |T \cap R^*|$ and $|T \cap B| = |T \cap B^*|$.
\item For every $i \in S \cap T$, we also have $\phi^{-1}(i) \subseteq T$. For every $i^* \in O \cap T$, we have $\phi(i^*) \in T$.
\item There is some facility $\oi \in T \cap S$  with $\deg(\oi) > 0$ designated as the {\em leader} that has the following properties.
Every other $i \in T \cap S - \{\oi\}$ is either good or very good and all good $i \in T \cap S - \{\oi\}$ have the same colour.
\end{itemize}
\end{lemma}

We focus on analyzing one block at a time to prove the approximation guarantee. This provides us with a cleaner way to describe the test
swaps and the additional structure will help handle the inevitable cases where we have to swap out some $i \in S$ but cannot swap
in all of $\phi^{-1}(i)$.
For example, this can happen if all blue facilities $i \in B$ have $\deg(i)$ being very large (so all $\deg(i') = 0$ facilities are red).
We will still need to close some of them in order to open facilities in $B^*$
when generating bounds via test swaps.


\subsection{Generating the Blocks} \label{sec:block}

We prove Lemma 1 in this section.
First, we describe how to partition $S \cup O$ into groups. These will then be combined to form the final blocks.
We say that a {\em group} is a subset $G$ of $S \cup O$ where $|G \cap S| = |G \cap O|$, there is exactly one $\oi \in G \cap S$ with $\deg(\oi) > 0$,
and $G \cap O = \phi^{-1}(\oi)$. Call this facility $\oi$ the {\em representative} of the group.

We classify groups $G$ in one of three ways.
\begin{itemize}
\item {\bf Balanced}: $|G \cap R| = |G \cap R^*|$ and $|G \cap B| = |G \cap B^*|$.
\item {\bf Good}: $\oi$ is a good facility and all other $i \in G \cap S - \oi$ have a different colour than $\oi$.
Note this means $|G \cap R| = |G \cap R^*| \pm 1$.
\item {\bf Bad}: $G$ is neither balanced nor good.
\end{itemize}

Note that here the good and bad are referring to groups, we emphasize that these are different than good and bad facilities.
Algorithm \ref{alg:alg1} describes a procedure for forming groups in a particular way that will be helpful in creating the final blocks.

\begin{algorithm}
\caption{Procedure for partitioning into groups}\label{alg:alg1}
\begin{algorithmic}
\State $S' \leftarrow S, O' \leftarrow O$
\While{$\exists$ some facility $i$ in $S'$ with $\deg(i) > 0$}
\State $G \leftarrow  \phi^{-1}(i) + i$
\If{$G \cup X$ is a balanced group for some $X \subseteq S'$}
\State $G' \leftarrow G \cup X$
\ElsIf{$G \cup X$ is a good group for some $X \subseteq S'$}
\State $G' \leftarrow G \cup X$
\Else
\State Let $X \subseteq \{i' \in S' : \deg(i') = 0\}$ such that $G \cup X$ is a bad group and\\ \hspace{12mm} all facilities in $S' - X$ have the same colour. \Comment{c.f. Lemma \ref{lem:lem1}}
\State $G' \leftarrow G \cup X$.
\EndIf
\State Output group $G'$ \Comment{$i$ is the representative of $G'$}
\State $S' \leftarrow S' - G', O' \leftarrow O'-G'$
\EndWhile
\end{algorithmic}
\end{algorithm}

\begin{lemma}\label{lem:lem1}
Each iteration correctly executes (i.e. succeeds in creating a group).
\end{lemma}
\begin{proof}
By a simple counting argument, there are always exactly
\[|O' - \{i \in S' : \deg(i) \neq 0\}|\]
very good facilities in $S'$. So we can always find a subset of very good facilities $X$
such that $G \cup X$ is a group. We prove that if the first two if conditions are false then we can find $X$ to ensure $S'-X$ only contains facilities of one colour.

There are 2 cases. Suppose $i$ is bad and, without loss of generality, that it is also red. Because we cannot extend $G$ to be a balanced group,
either there are less than $|\phi^{-1}(i) \cap \mB|$ very good blue facilities in $S'$ or less than $|\phi^{-1}(i) \cap \mR|-1$ very good red facilities in $S'$. In either case, first add
all very good facilities from the ``deficient'' colour to $X$ to use up that colour and then add enough very good facilities to $X$ of the other colour to ensure $|X| = |\phi^{-1}(i)|-1$.

In the other case when $i$ is good, we again assume without loss of generality that it is red.
Because we cannot form a good group, there are fewer than $|\phi^{-1}(i)|-1$ very good
blue facilities in $S'$. Use them up when forming $X$ and then add enough very good red facilities so that $|X| = |\phi^{-1}(i)| - 1$.
\end{proof}

Let $\mathcal G$ be the collection of groups output by Algorithm \ref{alg:alg1}.
We now show how to piece these groups together to form blocks.
\begin{algorithm}
\caption{Procedure for generating blocks}\label{alg:alg2}
\begin{algorithmic}
\State $\mathcal G' \leftarrow \mathcal G$
\While{there is some balanced group $G$ in $\mathcal G'$}
\State Output $G$ as a block with its own representative being the leader.
\State $\mathcal G' \leftarrow \mathcal G' - G$.
\EndWhile
\While{there are good groups $G, G' \in \mathcal G'$ with different coloured representatives}
\State Output the block $G \cup G'$ and choose either representative as the leader.
\State $\mathcal G' \leftarrow \mathcal G' - \{G,G'\}$
\EndWhile
\While{there is some bad group $G \in \mathcal G'$} \Comment{c.f. Lemma \ref{lem:app2}.}
\State Let $\mathcal G_0 \subseteq \mathcal G' - G$ consist only of good groups such that $G + \mathcal G_0$ is a block.
\State Output $G + \mathcal G_0$ with the representative of $G$ as the leader.
\State $\mathcal G' \leftarrow \mathcal G' - (G  + \mathcal G_0)$.
\EndWhile
\end{algorithmic}
\end{algorithm}

It is easy to verify that any ``block'' that is output by this algorithm indeed satisfies the properties listed in Lemma \ref{lem:block_prop}.

The following lemma explains why this procedure correctly executes and why all groups are used up. That is, it finishes the partitioning of $S \cup O$ into blocks.
For a union of groups $G^* = G_1 \cup \ldots \cup G_k$, define the {\em blue deficiency} of $G^*$ as $|G^*\cap B^*| - |G^* - B|$.
\begin{lemma}\label{lem:app2}
If $G$ is a bad group considered in some iteration of the last loop, we can find the corresponding $\mathcal G_0$ so that $G + \mathcal G_0$ is a block.
Furthermore, after the last while loop terminates then $\mathcal G' = \emptyset$.
\end{lemma}
\begin{proof}
Suppose, without loss of generality, that the blue very good facilities were used up the first time a bad group was formed in Algorithm \ref{alg:alg1}.
Thus, for every bad group $G' \in \mathcal G$ we have $|G \cap B| < |G \cap B^*|$.

Let $G$ be a group considered in some iteration of the last loop in Algorithm \ref{alg:alg2}.
As observed above, the blue deficiency of $G$ is strictly positive.

The blue deficiency of the union of all groups in $\mathcal G'$ is 0 because we have only removed
blocks from $\mathcal G'$ up to this point and, by definition, a block has blue deficiency 0. Thus, there must be some other group $G' \in \mathcal G'$ with strictly negative blue deficiency.
It cannot be that $G'$ is bad, otherwise it has nonnegative blue deficiency. It also cannot be that $G'$ is balanced or that it is a good group with a red representative,
because such blocks also have nonnegative blue deficiency.

Therefore, $G'$ must be good with a blue representative.
Good blocks with blue representatives have blue deficiency exactly -1. Add this $G'$ to $\mathcal G_0$. Iterating this argument with $G' + \mathcal G_0$, we add
good groups with blue representatives to $\mathcal G_0$ until $G' + \mathcal G_0$ is a block.
The layout in Figure \ref{fig:blocks} depicts a block constructed in this manner, the leftmost group in the figure is the bad group $G$ and the remaining groups
form $\mathcal G_0$.

After the last loop there are no groups with bad representatives or balanced groups.
Furthermore, all good groups must have the same colour of representative by the second loop.
If there were any good group, the blue deficiency of the union of groups in $\mathcal G'$ would then be nonzero, so there cannot be any good groups left.
That is, at the end of the last loop there are no more groups in $\mathcal G'$.
\end{proof}

\subsection{Standard Bounds}

Before delving into the analysis we note the following two bounds. The first has been used extensively in local search analysis and was first
proven in \cite{AGKMMP04} and the
second was proven in \cite{HKK12}. For convenience we will include the proofs here.
\begin{lemma}\label{lem:standard1}
For any $j \in C$, $d(j, \phi(o_j)) - c_j \leq 2c^*_j$.

\end{lemma}
\begin{proof}
By the triangle inequality and the definition of $\phi$, we have
\[ d(j,\phi(o_j)) \leq d(j, o_j) + d(o_j, \phi(o_j)) \leq c^*_j + d(o_j, s_j) \leq 2c^*_j + c_j. \]
\end{proof}

\begin{lemma}\label{lem:standard2}
For any $j \in C$, $d(j, \cent(\phi(o_j))) - c_j \leq 3c^*_j + c_j$.
\end{lemma}

\begin{proof}
By the triangle inequality and Lemma \ref{lem:standard1}, it suffices to prove $d(\phi(o_j), \cent(\phi(o_j))) \leq c^*_j + c_j$.
By definition of $\cent()$ and $\phi$,
\[ d(\phi(o_j), \cent(\phi(o_j))) \leq d(\phi(o_j), o_j) \leq d(s_j, o_j) \leq c^*_j + c_j. \]
\end{proof}

Finally, we often consider operations that add or remove a single item from a set. To keep the notation cleaner, we let $S + i$ and $S - i$ refer to $S \cup \{i\}$ and $S - \{i\}$, respectively,
for sets $S$ and items $i$.


\section{Multiswap Analysis} \label{sec:multiswap}
Recall that we are assuming $S = R \cup B$ is a locally optimum solution with respect to the heuristic that swaps at most $p$
facilities of each colour and that $O = R^* \cup B^*$ is some globally optimum solution.
We assume $p = t^2+1$ for some sufficiently large integer $t$.

Focus on a single block $T$.
For brevity, let $T^*_R = T \cap R^*$ and $T^*_B = T \cap B^*$ denote the red and blue facilities from the optimum solution
in $T$. Similarly let $T_R = T \cap R$ and $T_B = T \cap B$ denote the red and blue facilities from the local optimum solution
in $T$. The main goal of this section is to demonstrate the following inequality for group $T$

\begin{theorem}\label{thm:local_block}
For some absolute constant $\gamma$ that is independent of $t$, we have
\[ 0 \leq \sum_{j \in N^*(T^*_R \cup T^*_B)} \left[\left(1 + \frac{\gamma}{t}\right)c^*_j - c_j\right] +
\sum_{j \in N(T_R \cup T_B)} \left[\left(4 + \frac{\gamma}{t}\right) \cdot c^*_j + \frac{\gamma}{t}\cdot c_j\right]. \]
\end{theorem}

We first show the simple details of how this yields our main result.

\begin{pproof}{Theorem \ref{thm:main}}
Summing the inequalities stated in Theorem \ref{thm:local_block} over all blocks $T$, we see
\[ 0 \leq \sum_{j \in C} \left(5 + \frac{2\gamma}{t}\right) \cdot c^*_j - \left(1 - \frac{\gamma}{t}\right) \cdot c_j. \]
Multiplying through by $\frac{t}{t-\gamma}$, when $t > \gamma$ this shows
\[ 0 \leq \sum_{j \in C} \frac{5t + 2\gamma}{t-\gamma} \cdot c^*_j - c_j. \]
Rearranging, we see
\[ \cost(S) = \sum_{j \in C} c_j \leq \sum_{j \in C} \left(5 + \frac{7\gamma}{t-\gamma}\right) \cdot c^*_j = \left(5 + \frac{7\gamma}{t-\gamma}\right) \cdot \cost(S^*). \]
Recall that $p = t^2 + 1$ where $p$ is the number of swaps considered in the local search algorithm. Thus, Algorithm \ref{alg:alg_main} is a $(5+O(1/\sqrt p))$-approximation.
\end{pproof}

The analysis breaks into a number of cases based on whether $T^*_R$ and/or $T^*_B$ are large.
In each of the cases, we use the following notation and assumptions.
Let $\oi$ denote the leader in $T$. Without loss of generality, assume all other $i \in T_B \cup T_R$ with $\deg(i) > 0$
are blue facilities. Let $\overline B = \{i \in T_B - \oi : \deg(i) > 0\}$. Figure \ref{fig:blocks} illustrates this notation.

\begin{figure}
\includegraphics[width=1.0\textwidth]{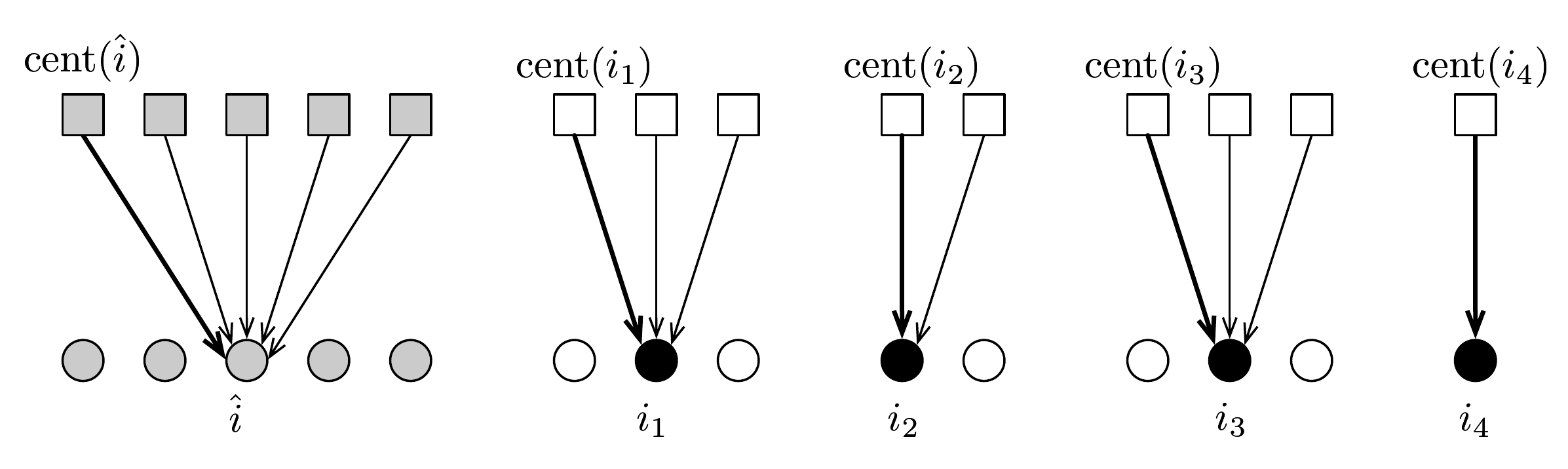}
\caption{Illustration of a block $T$. The facilities on the top are in $T \cap O$ and the facilities on the bottom are in $T \cap S$.
The directed edges depict $\phi$, and the thick edges connect $\cent(i)$ to $i$. The facilities coloured black lie in $\mB$, the
facilities coloured white lie in $\mR$, and the facilities coloured grey could either lie in $\mB$ or $\mR$. Note that $\overline B = \{i_1, i_2, i_3, i_4\}$.
The layout of the figure is suggestive of how the block was constructed by adding ``good'' groups to the initial bad group in Algorithm \ref{alg:alg2}} \label{fig:blocks}
\end{figure}

The swaps we consider in these cases are quite varied, but we always ensure we swap in $\cent(i)$ whenever some $i \in S \cap T$ with $\deg(i) > 0$ is swapped out.
This way, we can always bound the reassignment cost of each client $j$ by using either Lemma \ref{lem:standard1} or Lemma \ref{lem:standard2}.


\subsection{Case $|T^*_R| \leq t^2, |T^*_B| \leq t$}
In this case, we simply swap out all of $T_R \cup T_B$ and swap in all of $T^*_R \cup T^*_B$. Because $R \cup B$ is a locally optimum solution
and because this swaps at most $t^2$ facilities of each colour, we have.
\[ 0 \leq \cost(S \cup (T^*_R \cup T^*_B) - (T_R \cup T_B)) - \cost(S) \]
Of course, after the swap each client will move to its nearest open facility.
As is typical in local search analysis, we explicitly describe a (possibly suboptimal) reassignment of clients to facilities to upper bound this cost change.

Each $j \in N^*(T^*_R \cup T^*_B)$ is moved from $s_j$ to $o_j$ which incurs an assignment cost change of exactly $c^*_j - c_j$.
Each $j \in N(T_R \cup T_B) - N^*(T^*_R \cup T^*_B)$ is moved to $\phi(o_j)$. Note that $\phi(o_j) \not\in T$ so it remains open after the swap.
By Lemma \ref{lem:standard1}, the assignment cost change is bounded by $2c^*_j$.
Every other client $j$ that has not already been reassigned remains at $s_j$ and incurs no assignment cost change.
Thus,
\[ 0 \leq \sum_{j \in N^*(T^*_R \cup T^*_B)} (c^*_j - c_j) + \sum_{j \in N(T_R \cup T_B)} 2c^*_j \]
which is even better than what we are required to show for Theorem \ref{thm:local_block}.

We note that the analysis Section \ref{subsec:smallbig} could subsume this analysis (with a worse constant), but we have included it here
anyway to provide a gentle introduction to some of the simpler aspects of our approach.

\subsection{Case $|T^*_R| \geq t^2+1, |T^*_B| \geq t+1$}\label{subsec:bigbig}

We start by briefly discussing some challenges in this case. In the worst case, all of the $i_b \in T_B$ have $\deg(i)$ being very large. The issue here
is that we need to swap in each $i^*_b \in T^*_B$ in order to generate terms of the form $c^*_j - c_j$ for $j$ with $o_j = i^*_b$. But this requires
us to swap out some $i_b$. Since we do not have enough swaps to simply swap in all of $\phi^{-1}(i_b)$, we simply swap in $\cent(i_b)$.

Any client $j$ with $s_j$ being closed and $o_j \in \phi^{-1}(i_b) - \cent(i_b)$ cannot be reassigned to $\phi(o_j)$, so we send it to $\cent(\phi(o_j))$
and use Lemma \ref{lem:standard2} to bound the reassignment cost. This leaves a term of the form $+c_j$, so we have to consider additional
swaps involving $-c_j$ to cancel this out. These additional swaps cause us to lose a factor of roughly 5 instead of 3.

Another smaller challenge is that we do not want to swap out the leader $\oi \in T \cap S$ for a variety of technical reasons. However, since $|T^*_R|$ and
$|T^*_B|$ are both big, this is not a problem. When we swap out some $i^* \in T \cap O$, we will just swap in a randomly chosen facility in $T \cap S - \oi$ of the same colour.
The probability any particular facility is swapped in this way is very small. Ultimately, each facility in $T \cap S$ will be swapped out $2 + O(1/t)$ times in expectation.

To be precise, we partition the set of clients in $N(T_R \cup T_B)$ into two groups:
\[ C_{bad} := N(\overline B) \cap N^*(T^*_R - \cent(\overline B)) {\rm ~~~~and~~~~} C_{ok} := N(T_R \cup T_B - \oi) - C_{bad}. \]
We omit $N(\oi)$ from $C_{ok}$ because we will never close $\oi$ in this case.

The first group is dubbed {\em bad} because there may be a swap where both $s_j$ and $\phi(o_j)$ are closed yet $o_j$ is not opened
so we can only use Lemma \ref{lem:standard2} to bound their reassignment cost.
In fact, some clients $j \in C_{ok}$ may also be involved in such a swap, but we are able to use an averaging argument
for these clients to show that the resulting $+s_j$ term from using Lemma \ref{lem:standard2} appears with negligible weight
and does not need to be cancelled.




We consider the following two types of swaps to generate our initial inequality. 
\begin{itemize}
\item For each $i^*_b \in T^*_B$, choose a random $i_b \in T_B - \oi$. If $i_b \not\in \overline B$ (i.e. $\deg(i_b) = 0$) then simply swap out $i_b$ and swap in $i^*_b$.
If $i_b \in \overline B$ then swap out $i_b$ and a random $i_r \in T_R - \oi$ and swap in $i^*_b$ and $\cent(i_b)$.


\item For each $i^*_r \in T^*_R - \cent(\overline B)$, swap in $i^*_r$ and swap out a randomly chosen $i_r\in T_R - \oi$.
\end{itemize}
By choosing facilities at ``random'', we mean uniformly at random from the given set and this should be done independently for each invokation of the swap.

\begin{lemma}\label{lem:firststep}
\[ 0 \leq \sum_{j \in N^*(T^*_B \cup T^*_R)} \left( \frac{t+1}{t} \cdot c^*_j - c_j\right) +
\sum_{j \in C_{ok}} \left[\left(2 + \frac{5}{t}\right) c^*_j + \frac{1}{t} c_j\right]  + \frac{t+1}{t} \sum_{j \in C_{bad}} (3c^*_j + c_j). \]
\end{lemma}
\begin{proof}
For brevity, we will let $\beta_R = \frac{|T_R|}{|T_R-\oi|}$ and $\beta_B = \frac{|T_B|}{|T_B - \oi|}$. Note that $\beta_R, \beta_B \leq \frac{t+1}{t}$
and that either $\beta_R = 1$ or $\beta_B = 1$.

First consider a swap of the first type that swaps in $\{i^*_b, \cent(i_b)\}$ and swaps out $\{i_b, i_r\}$ for some $i_b$ with $\deg(i_b) > 0$.
Because $R \cup B$ is a local optimum the cost of the solution does not decrease after performing this swap. We provide
an upper bound on the reassignment cost.

Each $j \in N^*(\{i^*_b, \cent(i_b)\})$ is reassigned from $s_j$ to $o_j$ and incurs an assignment cost change of $c^*_j - c_j$.
Every client $j \in N(\{i_b, i_r\})$ that has not yet been reassigned is first moved to $\phi(o_j)$.
If this $\phi(o_j)$ remains open, assign $j$ to it. By Lemma \ref{lem:standard1}, the assignment cost for $j$ increases by at most $2c^*_j$.
If $\phi(o_j)$ is not open then $\phi(o_j) = i_b$ (because $\deg(i_r) = 0$) so we instead move $j$ to
$\cent(\phi(o_j)) = \cent(i_b)$. Lemma \ref{lem:standard2} shows the assignment cost increases by at most $3c^*_j + c_j$.
This can only happen if $s_j \in \{i_r, i_b\}$ and $\phi(o_j) = i_b$.

Combining these observations and using slight overestimates, we see
\begin{equation}\label{eqn:initial}
0 \leq \sum_{j \in N^*(\{i^*_b, \cent(i_b)\})} (c^*_j - c_j) + \sum_{\substack{j \in N(\{i_b, i_r\}) \\ \phi(o_j) \neq i_b}} 2c^*_j + \sum_{\substack{j \in N(\{i_b, i_r\}) \\ \phi(o_j) = i_b}} (3c^*_j + c_j).
\end{equation}

Now, if the random choice for $i_b$ in the swap has $\deg(i_b) = 0$, then swapping $\{i_b\}$ out and $\{i^*_b\}$ in
generates an even simpler inequality:
\begin{equation}\label{eqn:initial2}
0 \leq \sum_{j \in N^*(i^*_b)} (c^*_j - c_j) + \sum_{j \in N(i_b)} 2c^*_j.
\end{equation}
To see this, just reassign each $j \in N^*(i^*_b)$ from $s_j$ to $o_j$ and reassign the remaining $j \in N(i_b)$ from $s_j$ to $\phi(o_j)$
(which remains open because $\deg(i_b) = 0$) and use Lemma \ref{lem:standard1}.
%
%
%

Consider the expected inequality that is generated for this fixed $i^*_b$. We start with some useful facts
that follow straight from the definitions and the swap we just performed.
\begin{itemize}
\item Any $j \in N^*(\cent(\overline B))$ has $o_j$ open with probability $\frac{1}{|T_B-\oi|}$.
\item Any $j \in C_{bad}$ has $s_j$ being closed with probability $\frac{1}{|T_B-\oi|}$.
\item Any $j \in C_{ok} - N(T_R)$ has $s_j$ being closed with probability $\frac{1}{|T_B-\oi|}$. When this happens, if $o_j$ is not opened
then $\phi(o_j)$ must be open. \vspace{1mm}\\
That is, $j \in C_{ok}$ means $o_j \in T^*_B \cup \cent(\overline B)$. If $o_j \in T^*_B$ then $\phi(o_j) = \oi$ (by the structure of block $T$)
which remains open. If $o_j \in \cent(\overline B)$ then either $\phi(o_j)$ was not closed, or else $\cent(\phi(o_j)) = o_j$ was opened.
\item Any $j \in C_{ok} \cap N(T_R)$ has $s_j$ being closed with probability $\frac{|\overline B|}{|T_B-\oi|} \cdot \frac{1}{|T_R-\oi|}$.
If $o_j$ and $\phi(o_j)$ are closed, then we move $j$ to $\cent(\phi(o_j))$. However, this can only happen with probability
$\frac{1}{|T_B-\oi|}\cdot\frac{1}{|T_R - \oi|}$ since it must be that $\phi(o_j)$ is the blue facility that was randomly chosen to be closed.
\end{itemize}
Averaging \eqref{eqn:initial} over all random choices and using some slight overestimates we see
%
\begin{eqnarray*}
0 & \leq & \sum_{j \in N^*(i^*_r)} (c^*_j - c_j) + \frac{1}{|T_B - \oi|} \cdot \sum_{j \in N^*(\cent(\overline B))} (c^*_j - c_j) \\
& &  + \frac{1}{|T_B-\oi|} \left[\sum_{j \in C_{bad}} (3c^*_j + c_j) + \sum_{j \in C_{ok}-N(T_R)} 2c^*_j\right] \\
& & + \frac{1}{|T_B-\oi|}\cdot \frac{1}{|T_R-\oi|} \sum_{j \in C_{ok} \cap N(T_R)} (|\overline B| 2c^*_j + 3c^*_j + c_j).
\end{eqnarray*}
Summing over all $i^*_b \in N^*(T^*_B)$ shows
\begin{eqnarray}
0 & \leq & \sum_{j \in N^*(T^*_B)} (c^*_j - c_j) + \beta_B \cdot \sum_{j \in N^*(\cent(\overline B))} (c^*_j - c_j) \label{eqn:longer} \\
& & + \beta_B \cdot \left[\sum_{j \in C_{bad}} (3c^*_j + c_j) + \sum_{j \in C_{ok} - N(T_R)} 2c^*_j\right] + \frac{\beta_B}{|T_R-\oi|} \cdot \sum_{j \in C_{ok} \cap N(T_R)} ((2 |\overline B| + 3)c^*_j + c_j). \notag
\end{eqnarray}

Next, consider the second type of swap that swaps in some $i^*_r \in T^*_R - \cent(\overline B)$ and swaps out some randomly chosen $i_r \in T_R - \oi$.
Over all such swaps, the expected number of times each $i_r \in T_R-\oi$ is swapped out is $\frac{|T^*_R| - |\overline B|}{|T_R-\oi|} = \beta_R - \frac{|\overline B|}{|T_R-\oi|}$.
In each such swap, we reassign $j \in N^*(i^*_r)$ from $s_j$ to $o_j$ and every other $j \in N(i_r)$ from $s_j$ to $\phi(o_j)$ which is still open because $\deg(i_r) = 0$.
Thus,
\[ 0 \leq \sum_{j \in N^*(T^*_R - \cent(\overline B))} (c^*_j - c_j) + \left(\beta_R - \frac{|\overline B|}{|T_R -\oi|}\right) \cdot \sum_{j \in C_{ok} \cap N(T_R)} 2c^*_j \]
Scaling this bound by $\beta_B$, adding it to (\ref{eqn:longer}), and recalling $|T_R| \geq t^2$ shows
\begin{eqnarray*}
0 & \leq &  \sum_{j \in N^*(T^*_B)} (c^*_j - c_j) + \beta_B \cdot \sum_{j \in N^*(T_R)} (c^*_j - c_j) \\
& & + \beta_B \cdot \left[\sum_{j \in C_{bad}}(3c^*_j + c_j) + \sum_{j \in C_{ok} - N(T_R) } 2c^*_j +  \right]
+ \beta_B \cdot \sum_{j \in C_{ok} \cap N(T_R)} \left[\left(2\beta_R + \frac{3}{t^2}\right)\cdot c^*_j + \frac{1}{t^2}\cdot c_j\right].
\end{eqnarray*}
Recall that $\beta_B, \beta_R \leq \frac{t+1}{t}$ and also $\beta_B \cdot \beta_R \leq \frac{t+1}{t}$ to complete the proof of Lemma \ref{lem:firststep}.
\end{proof}

Our next step is to cancel terms of the form $+c_j$ in the bound from Lemma \ref{lem:firststep} for $j \in C_{bad}$. To do this, we again perform the second type of swap for each $i \in T^*_R - \cent(\overline B)$
but reassign clients a bit differently in the analysis.
\begin{lemma} \label{lem:cancel_1}
\[ 0 \leq \sum_{j \in C_{bad}} (c^*_j - c_j) + \frac{t+1}{t} \cdot \sum_{j \in C_{ok} \cap N(T_R)} 2c^*_j \]
\end{lemma}
\begin{proof}
For each $i^*_r \in T^*_R - \cent(\overline B)$, swap $i^*_r$ in and a randomly chosen $i_r \in T_r - \oi$. Rather than reassigning all $j \in N^*(i^*_r)$ to $i^*_r$, we only reassign
those in $C_{bad} \cap N^*(i^*_r)$. Since $\deg(i_r) = 0$ then any other $j \in N(i_r)$ can be reassigned to $\phi(o_j)$ and which increases the cost by at most $2c^*_j$.

Summing over all $i^*_r$, observing that $C_{bad} \subseteq T^*_R - \cent(\overline B)$, and also observing that each $j \in C_{ok}$ has
$s_j$ closed at most $\beta_R \leq \frac{t+1}{t}$ times in expectation, we derive the inequality stated in Lemma \ref{lem:cancel_1}.
\end{proof}

Adding the bounds stated in Lemmas \ref{lem:firststep} and \ref{lem:cancel_1} shows that Theorem \ref{thm:local_block} holds in this case.


\subsection{Case $|T^*_R| \geq t^2+1, |T^*_B| \leq t$}\label{subsec:bigsmall}

In this case, we start by swapping in all of $T^*_B$ and swapping out all of $T_B$ (including, perhaps, $\oi$ if it is blue). In the same swap, we also swap in $\cent(T_B)$
and swap out a random subset of the appropriate number of facilities in $T_R - \oi$. This is possible as $|T_R - \oi| \geq t \geq |\cent(T_B)|$.
By random subset, we mean among all subsets of $T_r-\oi$ of the necessary size, choose one uniformly at random.

As with Section \ref{subsec:bigbig}, we begin with a definition of bad clients that is specific to this case:
\[ C_{bad} := N(T_B) \cap N^*(T^*_R - \cent(T_B)). \]
Clients $j \in C_{bad}$ may be involved in swaps where both $s_j$ and $\phi(o_j)$ is closed yet $o_j$ is not opened
and we cannot make this negligible with an averaging argument.

\begin{lemma}\label{lem:init_smbig}
\begin{eqnarray*}
0 & \leq & \sum_{j \in N^*(T^*_B \cup \cent(T_B))} (c^*_j - c_j) + \frac{1}{t}\sum_{j \in N(T_R)} (3c^*_j + c_j)
 + \sum_{j \in C_{bad}} (3c^*_j + c_j)
\end{eqnarray*}
\end{lemma}
\begin{proof}
After the swap, reassign every $j \in N^*(T^*_B \cup \cent(T_B))$ from $s_j$ to $o_j$, for a cost change
of $c^*_j - c_j$. Every other $j$ that has $s_j$ being closed is first reassigned to $\phi(o_j)$. If this is not open,
then further move $j$ to $\cent(o_j)$ which must be open because the only facilities $i \in T_R \cup T_B$ with $\deg(i) > 0$ that were closed lie in $T_B$ and we opened $\cent(T_B)$.

If $j \in N(T_B) - C_{bad}$ then $o_j \in T^*_B \cup \cent(T_B)$ and we have already assigned $j$ to $o_j$. If $j \in C_{bad}$ then we have moved
$j$ to $\cent(\phi(o_j))$ and the cost change is $3c^*_j + c_j$ by Lemma \ref{lem:standard2}.

Finally, if $j \in N(T_R)$ then we either move $j$ to $\phi(o_j)$ or to $\cent(\phi(o_j))$ if $\phi(o_j)$ is not open.
The worst-case bound on the reassignment cost is $3c^*_j + c_j$ by Lemmas \ref{lem:standard1} and \ref{lem:standard2}. However, note that
$s_j \in T_R$ is closed with probability only $1/t$, since we close a random subset of $T_r-\oi$ of size at most $t$ and $|T_r - \oi| \geq t^2$.

%
\end{proof}

We still need to swap in $T^*_R - \cent(T_B)$. For each such facility $i^*_r$, 
swap in $i^*_r$ and swap out a randomly chosen $i_r \in T_R-\oi$.
The analysis of these swaps is essentially the nearly identical swaps in Section \ref{subsec:bigbig}, so we omit it and merely summarize what we get by combining the resulting inequalities
with the inequality from Lemma \ref{lem:init_smbig}. 

\begin{lemma}\label{lem:next_smbig}


\[ 
0 \leq \sum_{j \in N^*(T^*_R \cup T^*_B)} (c^*_j - c_j) + \sum_{j \in N(T_R)} \left(\frac{t^2+1}{t^2} \cdot 2c^*_j + \frac{1}{t} \cdot (3c^*_j + c_j)\right)
 + \sum_{j \in C_{bad}} (3c^*_j + c_j)
\]

%
\end{lemma}

We cancel the $+c_j$ terms for $j \in C_{bad}$ with one further collection of swaps.
For each $i^*_r \in T^*_R - \cent(T_B)$ we swap in $i^*_r$ and a randomly chosen $i_r \in T_R - \oi$.
The following lemma summarizes a bound we can obtain from these swaps. It is proven in essentially the same
way as Lemma \ref{lem:cancel_1}.
\begin{lemma}\label{lem:cancel}
\[ 0 \leq \sum_{j \in C_{bad}} (c^*_j - c_j) + \frac{t^2+1}{t^2} \cdot \sum_{j \in N(T_R)} 2c^*_j. \]
\end{lemma}
Adding this to the bound from Lemma \ref{lem:next_smbig} shows
\[ 0 \leq \sum_{j \in N^*(T^*_R \cup T^*_B)}(c^*_j - c_j) + \sum_{j \in N(T_R \cup T_B)} \left(\frac{t^2 + 3t + 1}{t^2} \cdot 4c^*_j + \frac{1}{t} \cdot c_j \right). \]


\subsection{Case $|T^*_R| \leq t^2, |T^*_B| \geq t+1$}\label{subsec:smallbig}

Because $\phi^{-1}(i) \subseteq T^*_R$ and $\deg(i) > 0$ for each $i \in \overline B$, then $|\overline B| \leq t^2$ as well.
We will swap all of $T^*_R$ for all of $T_R$, but we will also swap some blue facilities at the same time.
Let $B' = \overline B$ and let $\overline B'$ be an arbitrary subset of $T^*_B$ of size $|\overline B|$.

If $\oi \not\in T_R \cup B'$ then add $\oi$ to $B'$.
If $\cent(\oi) \not\in T^*_R \cup \overline B'$ then add $\cent(\oi)$ to $\overline B'$.
At this point, $\left| |\overline B'| - |B'| \right| \leq 1$ Add an arbitrary $i^*_b \in T^*_B - \overline B'$ to $\overline B'$ or $i_b \in T_B - B'$ to $B'$
to ensure $|\overline B'| = |B'|$.

We begin by swapping out $T_R \cup B'$ and swapping in $T^*_R \cup \overline B'$.
The following list summarizes the important properties of this selection, the first point emphasizes that this swap will not improve the objective function
since $S$ is a locally optimum solution for the $p$-swap heuristic where $p = t^2 + 1$.
\begin{itemize}
\item $|B'| = |\overline B'| \leq t^2+1$ and $|T^*_R| \leq t^2$.
\item $T^*_R$ was swapped in and $T_R$ was swapped out.
\item For each $i \in T_R \cup T_B$ with $\deg(i) > 0$, $i$ was swapped out and $\cent(i)$ was swapped in.
\end{itemize}

The following is precisely the clients $j$ that will be moved to $\cent(\phi(o_j))$ in our analysis.
\[ C_{bad} := [N(T_R \cup B') - N^*(T^*_R \cup \overline B')] \cap \{j : \phi(o_j) \in T_R \cup B'\}. \]
As before, define $C_{ok} = N(T_R \cup T_B) - C_{bad}$.

The following bound is generated from swapping out $T_R \cup B'$ and swapping in $T^*_R \cup \overline B'$
and follows from the same arguments we have been using throughout the paper.
\begin{lemma}\label{lem:nearly}
\[0 \leq \sum_{j \in N^*(T^*_R \cup \overline B')} (c^*_j - c_j) + \sum_{j \in C_{ok} \cap N(T_R \cup B')} 2c^*_j + \sum_{j \in C_{bad}} (3c^*_j + c_j) \]
\end{lemma}

Next, let $\kappa_B : (T^*_B - \overline B') \rightarrow (T_B - B')$ be an arbitrary bijection of the remaining blue facilities that were not swapped.
For every $i^*_b \in T^*_B - \overline B'$, consider the effect of swapping in $i^*_b$ and swapping out $\kappa_B(i^*_b)$. Note that every facility $i_b$ swapped out in this way has $\deg(i_b) = 0$.
So we can derive two possible inequalities from such swaps. 
\begin{equation}\label{eqn:blue1}
0 \leq \sum_{j \in N^*(i^*_b)} (c^*_j - c_j) + \sum_{j \in N(\kappa_B(i^*_b))} 2c^*_j 
\rm{~~~and~~~}
0 \leq \sum_{j \in N^*(i^*_b) \cap C_{bad}} (c^*_j - c_j) + \sum_{j \in N(\kappa_B(i^*_b))} 2c^*_j.
\end{equation}
The second inequality follows from only reassigning clients $j \in N^*(i^*_b) \cap C_{bad}$ from $s_j$ to $o_j$.

Adding the bound in Lemma \ref{lem:nearly} to the sum of both inequalities in 
over all $i^*_b \in T^*_B - \overline B'$
and noting that $\kappa_B(T^*_B - \overline B) \cap (T_R \cup B') = \emptyset$, we see
\[ 0 \leq \sum_{j \in N^*(T^*_R \cup T^*_B)} (c^*_j - c_j) + \sum_{j \in N(T_R \cup T_B)} 4c^*_j. \]


\section{Locality Gaps}\label{sec:localgap}
Here we prove Theorem \ref{thm:localgap}. Let $p,\ell$ be integers satisfying $p \geq 1$ and $\ell \geq 2p$.
Consider the instance with $k_r = p+1$ and
$k_b = p(\ell+1)$ depicted in Figure \ref{fig:localgap}. Here, $\beta = 2p$ and $\alpha = \beta \cdot (\ell-p)$.

\begin{figure}
\includegraphics[width=1.0\textwidth]{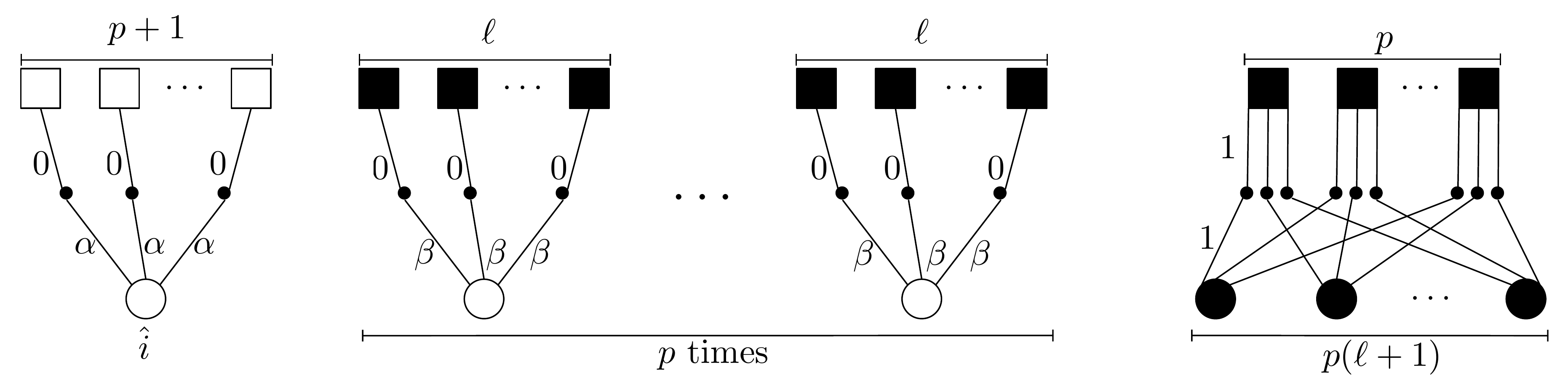}
\caption{Illustration of the bad locality gap. Blue facilities are depicted with black and red facilities are depicted with white.
The top facilities are the global optimum and the bottom are the local optimum (all of $\mR$ and $\mB$
is depicted in the picture).
Each client is represented by a small black dot.
The metric is the shortest path metric of the presented graph, if two locations are not connected in the picture then their
distance is a very large value.
Every edge in the right-most group with $p^2(\ell+1)$ clients has length 1.
Recall $\beta = 2p$ and $\alpha = (\ell-p) 2p$.} \label{fig:localgap}
\end{figure}

The cost of the local optimum solution is $\alpha \cdot (p+1) + \beta \cdot p\ell + p^2(\ell+1)$ and the cost of the global optimum solution
is simply $p^2(\ell+1)$. Through some careful simplification, we see the local optimum solution has cost at least $5 + \frac{2}{p} - \frac{10p}{\ell+1}$
times the global optimum solution.

To complete the proof of Theorem \ref{thm:localgap}, we must verify that the presented local optimum solution
indeed cannot be improved by swapping up to $p$ facilities of each colour.

We verify that the solution depicted in Figure \ref{fig:localgap} is indeed a locally optimum solution.
Suppose $0 \leq R \leq p$ red facilities and $0 \leq B \leq p$ blue facilities are swapped. We break the analysis into four simple cases.

In what follows, we refer to the leftmost collection of only red facilities in Figure \ref{fig:localgap}
as the {\em left group}, the rightmost collection of only blue facilities
as the {\em right group}, and the remaining facilities as the {\em middle group}. We
also let the term {\em subgroup} refer to one of the $p$ smaller collections of facilities in the middle group.
In each case, let $B' \leq B$ denote the number of global optimum facilities from the middle group that are swapped in. Recall that
$\oi$ denotes the local optimum facility in the left group.

\subsection*{Case $R = 0$}

The only clients that can move to a closer facility are the $B'$
clients in the middle group that have their associated optimum facilities swapped in. Also, precisely $B \cdot (p-B+B')$ facilities in the right
group are not adjacent to any open facility so their assignment cost increases by 2.

Overall, the assignment cost change is exactly $2B \cdot (p-B+B') - \beta B'$.
As $\beta = 2p$ and $B \leq p$, this quantity is minimized at $B' = B$ leaving us with a cost change of
$2Bp - \beta B = 0$. So, if $R = 0$ then no choice of blue facilities leads to an improving swap.

\subsection*{Case $R \geq 1$ and $\oi$ is not swapped out.}
In the left group, precisely $R$ clients move to their close facility and the total savings is $-\alpha R$. In the middle group,
precisely $B'$ clients move to their close facility and the total savings is $-\beta B'$.

In fact, it is easy to see that the cheapest such swap occurs when the $B' \leq p \leq \ell$ facilities in the middle group that are swapped in
are part of subgroups where the local optimum facility is swapped out (which is why we assume $R \geq 1$). 
So, there are exactly\ $R\ell - B'$ other clients
$j$ where both $o_j$ and $s_j$ are closed and each pays an additional $\geq \beta$ to be connected. Finally, the right group pays an additional
$2B(p-B+B')$ to be connected.

Overall, the cost increases by $2B(p-B+B') + \beta (R\ell-2B') - \alpha R$. As $2B - 2\beta = 2B - 4p \leq -2p$,
this is minimized at $B' = B$. The cost change is then $2Bp + \beta (R\ell - 2B) - \alpha R$. Recall that $\alpha = (\ell-p) \beta < \ell \beta$,
so this is, in turn, minimized when $R = 1$.

Reducing further, the cost change is $2Bp + \beta \ell - 2B\beta - \alpha$. Setting $B = p$ to maximize, the change is
$2p^2 + 2p\ell - 4p^2 - (\ell-p)2p = 0$. So, no swap that swaps at least one red facility but not $\oi$ can find a cheaper solution.

\subsection*{Case $R = 1$ and $\oi$ is swapped out.}
The cost change in the left group is $(p-1) \alpha \geq 0$ since $p$ clients must move an additional $\alpha$ and only one client saves
$\alpha$. The cost change from the remaining groups is the same as in the first case $R = 0$, so the overall assignment cost does not decrease.

\subsection*{Case $R \geq 2$ and $\oi$ is swapped out.}
The cost change in the first group is exactly $(p+1-2R) \alpha$. Similar to the second case, the cost change in this case is minimized
when each subgroup that has its local optimum facility closed also has one of its global optimum facility opened, and all $B'$ facilities opened in the middle group belong to a subgroup
having its local optimum closed.

The cost change is then $2B(p-B+B') + \beta ((R-1)\ell - 2B') + \alpha (p+1-2R)$. Again, this is minimized at $B' = B$
which yields a cost change of $2Bp + \beta((R-1)\ell - 2B) + \alpha(p+1-2R)$. Now, $\beta\ell \leq 2\alpha$ because $\ell \geq 2p$, so
this is further minimized at $R = p$ and the cost increase is $2Bp + \beta((p-1)\ell -2B) - \alpha(p-1)$. Again, setting $B = p$ to minimize
the cost change we see it is $2p^2 + \beta((p-1)\ell - 2p) - \alpha(p-1)$.

Expanding with $\beta = 2p$ and $\alpha = 2p(\ell-p)$, the cost change finally seen to be
\[ -2p^2 + 2p(p-1)\ell - 2p(\ell-p)(p-1) = 2p^3 - 4p^2. \]
The last expression is nonnegative for $p \geq 2$.

\subsection{Summarizing}
No matter which $\leq p$ red and $\leq p$ blue facilities are swapped, the above analysis shows the assignment cost does not decrease.
The only potentially concerning aspect is that the very last case derived an inequality that only holds when $p \geq 2$. Still, this analysis
does apply to the single-swap case (i.e. $p = 1$) since the last case with $R \geq 2$ does not need to be considered when $p = 1$.


\section{Conclusion}\label{sec:conclusion}
We have demonstrated that a natural $p$-swap local search procedure for \rbm is a $(5 + O(1/\sqrt p))$-approximation.
This guarantees a better approximation ratio than the single-swap heuristic from \cite{HKK12}, which we showed
may find solutions whose cost is $(7-\epsilon) \cdot OPT$ for arbitrarily small
$\epsilon$. Our analysis is 
essentially tight in that the $p$-swap heuristic may find solutions whose cost is $(5+\frac{2}{p} - \epsilon) \cdot OPT$.

More generally, one can ask about the $p$-swap heuristic for the generalization
where there are many different facility colours. If the number of colours is part of the input
then any local search procedure that swaps only a constant number of facilities in total cannot provide good approximation guarantees \cite{KKNSS11}. However, if the number of different colours
is bounded by a constant, then perhaps one can get better approximations through multiple-swap heuristics.

However, generalizing the approaches taken with \rbm to this setting seems more difficult;
one challenge is that it is not possible to get such nicely structured blocks.
It would also be interesting to see what other special cases of \mm admit good local-search based approximations.
For example, the \textsc{Mobile Facility Location} problem studied in \cite{AFS13} is another special case of \mm
that admits a $(3+\epsilon)$-approximation through local search.

Finally, the locality gap of the $p$-swap heuristic for \km is known to be $3 + \frac{2}{p}$ \cite{AGKMMP04}
and we have just shown it is at least $5 + \frac{2}{p}$ for \rbm. Even if the multiple-swap heuristic
for the generalization to a constant number of colours can provide a good approximation, this constant may be worse
than the alternative 8-approximation obtained through Swamy's general \mm approximation \cite{S14}.

\end{document}